\documentclass[11pt]{article}
\usepackage[utf8]{inputenc}
\usepackage{macros}
\usepackage[
    backend=biber,
    style=alphabetic,
    sorting=nyt,
    backref=true,
    maxbibnames=9
]{biblatex}

\newcommand{\CRS}{C_{\mathrm{RS}}}

\newcommand{\tildeM}{\widetilde{M}}
\newcommand{\tildeP}{\widetilde{P}}

\title{Error-Correcting Graph Codes}

\newtoggle{blind}
\togglefalse{blind}

\author{
    \iftoggle{blind}{}{
        Swastik Kopparty\thanks{Department of Mathematics and Department of Computer Science, University of Toronto.
            Research supported by an NSERC Discovery Grant.
        Email: swastik.kopparty@utoronto.ca} \and
        Aditya Potukuchi\thanks{Department of Computer Science, York University. Supported in part by NSERC Discovery grants RGPIN-2023-05087 \& DGECR-2023-00408 Email: apotu@yorku.ca} \and
        Harry Sha\thanks{Department of Computer Science, University of Toronto. Email: shaharry@cs.toronto.edu}
    }
}

\date{\today} 

\begin{document}

\maketitle

\begin{abstract}
    In this paper, we construct {\em Error-Correcting Graph Codes}. An error-correcting graph code of distance $\delta$ is a family $C$ of graphs on a common vertex set of size $n$, such that if we start with any graph in $C$, we would have to modify the neighborhoods of at least $\delta n$ vertices in order to obtain some other graph in $C$.
This is a natural graph generalization of the standard Hamming distance error-correcting codes for binary strings. 

Yohananov and Yaakobi were the first to construct codes in this metric, constructing good codes for $\delta < 1/2$, and optimal codes for a large-alphabet analogue. We extend their work by showing
\begin{enumerate}
    \item Combinatorial results determining the optimal rate vs. distance trade-off nonconstructively.
    \item Graph code analogues of Reed-Solomon codes and code concatenation, leading to positive distance codes for all rates and positive rate codes for all distances.
    \item Graph code analogues of dual-BCH codes, yielding large codes with distance $\delta = 1-o(1)$. This gives an explicit ``graph code of Ramsey graphs''.
\end{enumerate}

Several recent works, starting with the paper of Alon, Gujgiczer, Körner, Milojević, and Simonyi, have studied more general graph codes; where the symmetric difference between any two graphs in the code is required to have some desired property. Error-correcting graph codes are a particularly interesting instantiation of this concept.
\end{abstract}

\newpage
\tableofcontents
\newpage

\section{Introduction}
In this paper, we study {\em Error-Correcting Graph Codes}.
These are large families of undirected graphs on the same vertex set such that any two graphs in the family are ``far apart'' in a natural graph distance. Informally, the graph distance between two graphs on the same vertex set of size $n$ measures the minimum number of vertices one needs to delete to make the resulting graphs identical (not just isomorphic). 
This can also be thought of as (1) the number of vertices whose neighborhoods one has to modify to go from one graph to another, (2) the vertex cover number of the symmetric difference of the two graphs, or (3) $n$ minus the largest independent set in the symmetric difference of the two graphs.

\noindent{\bf Definition: (Graph Distance)\quad} Given two graphs $G$ and $H$ on vertex set $[n]$, the graph distance $\dgraph(G,H)$ is the size of the smallest set $S \subseteq [n]$ such that $G[\overline{S}] = H[\overline{S}]$.

This is a very natural metric and encompasses deep information about graphs. 
For example, note the following two simple facts: (1) the graph distance of a graph from the empty graph is $n$ minus the independence number of the graph. 
(2) the graph distance of a graph from the complete graph is $n$ minus the clique number of the graph. 
Thus, the answer to the question: ``how far can a graph be from both the empty graph and the complete graph?'' 
is precisely the question of finding the right bound for the diagonal Ramsey numbers; the answer is $n - O(\log n)$.

Codes in the graph-distance metric were defined and initially studied by Yohananov and Yaakobi~\cite{yohananovCodesGraphErasures2019}, and Yohanonanov, Efron, and Yaakobi \cite{yohananovDoubleTripleNodeErasureCorrecting2020}.
Their setting allows for arbitrary symbols $\alpha \in \F_q$ to be written on every edge of $K_n$ ($q = 2$ corresponds to graph codes), including self-loops. They gave a complete understanding of the combinatorial limits for large $q$, optimal explicit constructions in the large $q$ setting, and asymptotically good explicit constructions for the $q=2$ setting for distance at most $n/2$. 

Error-correcting graph codes also fall into the general framework of graph codes defined by Alon, Gujgiczer, Körner, Milojević, and Simonyi~\cite{AGKMS-graph-codes}, where for a fixed family $\mathcal F$ of graphs, one seeks a large code $C$ of graphs on the same $n$-vertex set such that the symmetric difference of any two graphs in $C$ does not lie in $\mathcal F$. 
This class of problems was studied for a wide variety of natural $\mathcal F$ in a number of recent works~\cite{AGKMS-graph-codes,AlonGC1,AlonGC2}. 
As discussed in~\cite{AlonGC1}, for a suitable choice of $\mathcal F$, graph codes become equivalent to classical Hamming distance codes.

Additionally, the notion of graph distance arises in the definition of node differential privacy (see, for example \cite{chenRecursiveMechanismNode2013, hayAccurateEstimationDegree2009, kasiviswanathanAnalyzingGraphsNode2013, jainTimeAwareProjectionsTruly2024}). 
One instantiation of this setting is a graph encoding a social network where vertices correspond to people and edges correspond to social connections.
The goal for node differential privacy is to design an algorithm $\mathcal{A}$ that approximately computes certain statistics of the graph (such as counting edges, triangles, and connected components) while maintaining each individual’s privacy. 
Here, privacy is ensured by requiring that the output distribution on a certain graph does not change by much when any one vertex is deleted.
In other words, for any graphs $G$, $H$ of graph distance $1$, the output distribution of $\mathcal{A}$ on the $G$ and $H$ should be similar.
Graph distances greater than $1$ are then considered in the continual release model where the graph varies over time as studied in \cite{jainTimeAwareProjectionsTruly2024}.

Finally, we note that error-correcting codes are pseudorandom objects, and the connection to Ramsey graphs suggests that error-correcting graph codes might be closely related to pseudorandom graphs. 
Thus, the problem of studying and explicitly constructing a pseudorandom family of pseudorandom graphs is interesting in its own right.


\subsection{Rate versus distance and prior work}

Similar to the Hamming setting, we briefly define the dimension, rate, and distance of a graph code $C$ on $n$ vertices where each edge is allowed to have an arbitrary symbol $\alpha \in \mathbb{F}_q$ written on it. 
\begin{itemize}
    \item \textbf{Dimension.} $k = \log_q(|C|)$.
    \item \textbf{Rate.} $R = \log_q(|C|) / \binom{n}{2}$.
    \item \textbf{Distance.} The distance of a code is the largest $d$ such that $\dgraph(G, H) \geq d$ for each $G, H \in C$ such that $G \neq H$. The relative distance $\delta$ is $d / n$.
\end{itemize}

Unless specified otherwise, we will always be interested in the asymptotics of the above parameters as $n \to \infty$. All construction of graph codes in this paper will be linear subspaces of $\F_q^{{n \choose 2}}$. A construction of graph codes is called explicit if a basis for the space can be computed in time $\poly(n)$. A construction of graph codes is called strongly explicit if the $(i,j)$ entry of the $e$'th basis element can be computed in time $\poly(\log(n))$ when given $i,j,e$ as input.

Yohananov and Yaakobi~\cite{yohananovCodesGraphErasures2019} showed, by an argument similar to the Singleton bound, that for a graph code of dimension $k$ and distance $d$, we have

\begin{equation}
k \leq \binom{n - d + 1}{2}.
\end{equation}

In terms of and rate $R$, and relative distance $\delta$,  we have

\begin{equation}
\label{eq:sing}
    R \leq (1 - \delta)^2 + o(1).
\end{equation}


They~\cite{yohananovCodesGraphErasures2019} (Construction 3) then showed, by explicit construction, that this bound is tight for large $q$. The main ingredient is a way to convert Hamming distance into graph distance using the tensor product. 
A similar code will be an important ingredient in our constructions.

\cite{yohananovCodesGraphErasures2019} (Construction 5) also gave explicit constructions of codes over the binary alphabet for $\delta \leq 1/2$.
This leverages a connection to symmetric array codes and a construction by Schmidt \cite{schmidtSymmetricBilinearForms2015} related to the rank metric. The trade-off achieved, $R \sim 1 - 2 \delta$, is close to the upper bound~\eqref{eq:sing} for $\delta$ very small.

 Finally, we note that \cite{yohananovCodesGraphErasures2019} and \cite{yohananovDoubleTripleNodeErasureCorrecting2020} also give optimal constructions in the $\delta = O(1/n)$ regime.
 
 These constructions\footnote{Note that all of the constructions in \Cref{table:prev-constructions} are for undirected graphs where self-loops are allowed. 
Since these are linear codes, one can impose $n$ linear constraints to get codes with no self-loops (zeros on the diagonals of the adjacency matrix). 
Since the block length of these codes is $\binom{n}{2}$, adding these these constraints changes the rate by $o(1)$ while preserving explicitness.

Adding the $0$-diagonal constraint to the symmetric array codes while preserving {\em strong explicitness} is more subtle. Alessandro Neri~\cite{Neri-pc} observed that this can be done using results of Gow and Quinlan (Theorem 6 and Theorem 7 of~\cite{GowQuinlan}), which themselves are related to results of Delsarte and Goethals~\cite{DG-alternating}.}
 are summarized in \Cref{table:prev-constructions}.

\begin{center}
    \begin{table}[]
    \begin{tabular}{llll}
\textbf{Name}                                                          & \textbf{$\delta$} & \textbf{Field}     & \textbf{Tradeoff} \\ \hline
\cite{yohananovCodesGraphErasures2019}               & $\leq 1$          & $\F_q, q \geq n-1$ & $R = (1-\delta)^2$ (optimal)           \\
\cite{yohananovCodesGraphErasures2019}           & $\leq 1/2$        & $\F_2$             & $R = 1 - 2\delta$ \\
\cite{yohananovCodesGraphErasures2019},\cite{yohananovDoubleTripleNodeErasureCorrecting2020}                & $= O(1/n)$           & $\F_2$             & optimal          \\
\end{tabular}
    \caption{
        Summary of constructions in \cite{yohananovCodesGraphErasures2019}, and \cite{yohananovDoubleTripleNodeErasureCorrecting2020}
    }
    \label{table:prev-constructions}
    \end{table}
\end{center}

\subsection{Results}

Our main results are:
\begin{enumerate}[1.)]
    \item We observe that there are binary graph codes achieving $R = (1-\delta)^2 - o(1)$ for any constant $\delta \in (0,1)$. Thus $R = (1-\delta)^2 - o(1)$ (or $\delta=  1-\sqrt{R} -o(1)$) is the optimal $R$ vs $\delta$ tradeoff even for the binary alphabet (\cite{yohananovCodesGraphErasures2019} showed this for large alphabet).
    \item We give constructions of graph codes that have positive constant $R$ for all constant $\delta < 1$. 
        We give an explicit construction with $\delta = 1 - O(R^{1/6})$, and a strongly explicit construction with $\delta = 1-O(R^{1/8})$.
        We also give a quasi-polynomial time explicit construction achieving $\delta = 1 - O(R^{1/4})$.
        Compare this with the optimal nonconstructive tradeoff of $\delta = 1 - R^{1/2}$. 

        Although these codes are not optimal, they are the first binary error-correcting graph codes achieving $\delta > 1/2$ with a constant rate.
    \item We give (strongly) explicit constructions of graph codes with very high distance $\delta = 1 - O(n^{-\epsilon})$ and inverse polynomial rate $R = \Omega(n^{\epsilon - 1/2})$ for constant $\epsilon > 0$. 
        This gives a ``graph code of Ramsey graphs,’’ as will be discussed later.
\end{enumerate}

\paragraph{Independent work:} Pat Devlin and Cole Franks~\cite{Pat-comm} independently proposed the study of graph error-correcting codes under this metric, determined the optimal $R$ vs $\delta$ tradeoff for binary codes, and gave some weaker explicit constructions of graph codes that worked for certain ranges of $R$ and $\delta$.

\subsection{Techniques}
We now discuss our techniques. We often specify graphs by their adjacency matrices, viewed as matrices with $\F_2$ entries.

Our nonconstructive existence result is a straightforward application of the probabilistic method. 
We consider a uniformly random $\F_2$-linear subspace of the $\F_2$-linear space of symmetric $0$-diagonal $n\times n$ matrices (i.e., the space of all adjacency matrices of graphs); this turns out to give a good graph code with optimal $R$ vs $\delta$ tradeoff.

To get explicit constructions of asymptotically good codes for any constant $\delta \in (0,1)$, we use several ideas.
\begin{enumerate}
    \item We start with a slight variation of \cite{yohananovCodesGraphErasures2019} (Construction 3), which gives a way to get a good graph code from a classical Hamming-distance linear code $C \subseteq \F_2^n$. 
        We first consider the tensor code $C \otimes C$, where the elements are matrices whose rows and columns are codewords of $C$. 
        A-priori, $C$ could contain matrices that are neither symmetric nor have a $0$ diagonal.
        But interestingly, if we consider the 
    set $C^*$ of all matrices in $C \otimes C$ that are symmetric and have $0$ diagonal, then $C^*$ is a linear space with quite large dimension. 
        In particular, if the classical Hamming distance code $C$ has positive rate, then so does the graph code $C^*$. 
        We call this construction $C^* = \STCZD(C)$ (Symmetric Tensor Code with Zero Diagonals).

        It turns out that if $C$ has good relative distance (in the Hamming metric), then $\STCZD(C)$ has good distance in the graph metric. 
        However, the relative graph distance of $\STCZD(C)$ such a code is bounded by the relative distance of $C$ -- and since $C$ is a binary code, this is at most $1/2$. 
    
    \item Now, we bring in another idea from the Hamming code world: code concatenation. 
        Instead of constructing a graph code of symmetric zero-diagonal matrices over $\F_2$, we instead construct a ``large-alphabet graph code" of symmetric zero-diagonal matrices over $\F_q$  for some large $q = 2^t$ and then try to reduce the alphabet size down to $2$ by replacing the $q$-ary symbols with $\F_2$-matrices with suitable properties. 
        
        Applying the analog of $\STCZD$ to a large alphabet code allows one to get large-alphabet graph codes with large $\delta$, approaching $1$ (since over large alphabets, Hamming distance codes can have length approaching $1$). 
        Using Reed-Solomon codes as these large alphabet codes also allows us to make the $\STCZD$ construction strongly explicit. 
        Furthermore, when applied to Reed-Solomon codes, these codes have a natural direct description: these are the evaluation tables of low-degree bivariate polynomials $P(X, Y)$ on product sets $S \times S$ that are (1) symmetric (to get a symmetric matrix), and (2) multiples of $(X-Y)^2$ (to get zero diagonal).
        
    \item What remains now is to develop the right kind of concatenation so that the resulting graph code has good distance. 
        This turns out to be subtle and requires an ``inner code" with a stronger ``directed graph distance" property with $\delta$ nearly 1. 
        Fortunately, this inner code we seek is very slowly growing in size, and we may find it by brute force search.  
        This concludes our description of our explicit construction of graph codes with $\delta$ approaching $1$ and positive constant $R$.
\end{enumerate}

Finally, we discuss our constructions for very high distances, $\delta = 1 - o(1)$.
In this regime, as mentioned earlier, this is related to constructions of Ramsey graphs, a difficult problem in pseudorandomness with a long history. 
Our constructions work up to $\delta = 1 - \Omega\left(\frac{1}{\sqrt{n}}\right)$;  concretely, we get a large linear space of graphs such that all graphs in the family have no clique or independent set of size $\Omega(\sqrt{n})$. 
The construction is based on polynomials over finite fields of characteristic $2$: When $n = 2^t$, we consider a linear space of certain low degree univariate polynomials $f(X)$ over $\F_{2^t}$ and create the $\F_2$ matrix with rows and columns indexed by $\F_{2^t}$ whose $x,y$ entry is $\Tr(f(x+y))$. 
Here $\Tr$ is the finite field trace map from $\F_{2^t}$ to $\F_2$. 
The use of $\Tr$ of polynomials is inspired by the construction of dual-BCH codes. 
We then show that any such matrix has no large clique or independent set unless $\Tr \circ f$ is identically $0$ or identically $1$ (corresponding to the empty and complete graphs, respectively). 
The proof uses the Weil bounds on character sums and a Fourier analytic approach to bound the independence number for the graphs.

Our constructions are listed in \Cref{table:list}. 

\begin{center}
    \begin{table}[h]
        \begin{tabular}{lll}
        Name                                            & Approximate Tradeoffs                                                      & Explicit/Strongly Explicit?  \\ \hline
        Random Linear Codes (\Cref{prop:ub}) & $R = (1 - \delta)^2 - o(1)$                                             & No/No               \\
        Concatenated RS Tensor Codes (\Cref{code:concat-RS})   & $R = (1 -\sqrt{\delta})^4 - o(1)$                                                 & Almost Yes\footnote{This construction has a quasipolynomial time construction $n^{O(\log n)}$.}/No                \\
        Double Concatenated  RS Tensor Codes  & $R = (1 - \delta^{1/3})^6 - o(1)$                                                 & Yes/No                \\
        Triple Concatenated  RS Tensor Codes (\Cref{code:triple-concatenation})   & $R = (1 - \delta^{1/4})^8 - o(1)$                                                 & Yes/Yes                \\
        Dual BCH Codes (\Cref{code:dualBCH})                        & $R = \log(n)(1 - \delta) / \sqrt{n}$                                       & Yes/Yes              
        \end{tabular}
        \caption{
            A list of constructions of error-correcting graph codes in this paper. 
            A particular focus is the regime where $\delta$ is very close to $1$.
        } 
        \label{table:list}
    \end{table}
\end{center}


\subsection{Concluding thoughts and questions}

The most interesting question in this context is obtaining explicit graph code constructions with optimal $R$ vs $\delta$ tradeoff. 
While we have several constructions achieving nontrivial parameters in various regimes, it would even be interesting to get the right asymptotic behavior for the endpoints with $\delta$ approaching $1$.
The setting of large $\delta$ (including $\delta = 1 -o(1)$) seems especially challenging, given the connection with the notorious problem of constructing Ramsey graphs.

Another interesting question is to get decoding algorithms for graph codes. 
For a certain graph code $C$, if we are given a graph that is promised to be close in graph distance to some graph $G$ in $C$. 
Then, can we efficiently find $G$?

A more general context relevant to error-correcting graph codes is the error correction of strings under more general error patterns. 
Suppose we have a collection of subsets $S_i \subseteq [m]$ for $i \in [t]$, where $\bigcup_i S_i = [m]$. 
These $S_i$ denote the corruption zones; a single ``corruption" of a string $z \in \{0,1\}^m$ entails, for some $i \in [t]$, changing $z|_{S_i}$ to something arbitrary in $\{0,1\}^{S_i}$. 
We want to design a code $C \subseteq \{0,1\}^m$ such that starting at any $x \in C$ if we do fewer than $d$ corruptions to $x$, we do not end up at any $y \in C$ with $y \neq x$. 
When the $S_i$ are all of size $b$ and form a partition of $[m]$ into $t = m/b$ parts, then such a code is exactly the same as a classical Hamming distance code an alphabet of size $2^b$. 
Error-correcting graph codes give a first step into the challenging setting where the $S_i$ all pairwise intersect - here we have $m = {\binom{n}{2}}$, $t = n$, the $S_i$ (which correspond to all edges incident on a given vertex) all have size $n-1$, and every pair $S_i$ and $S_j$ intersect in exactly $1$ element. 
It would be interesting to develop this theory -- to both find the limits of what is achievable and to develop techniques for constructing codes against this error model.

Finally, there are many other themes from classical coding theory that could make sense to study in the context of graph codes and graph distance, including in the context of sublinear time algorithms. 
It would be interesting to explore this.

\paragraph{Organization of this paper:} 
We set up basic notions in Section~\ref{sec:basics}.
We show the existence of optimal graph codes in Section~\ref{sec:combinatorics}.
In \Cref{sec:concatenation} we construct asymptotically good codes.
Finally, in Section~\ref{sec:weil}, we show explicit constructions of graph codes with very high distance.

\section{Graph codes: Basics}\label{sec:basics}

\paragraph{Definitions and notations.}
All graphs will be simple, undirected graphs on the vertex set $[n]$ unless otherwise noted.
For any graph $G$, use $A_G$ to denote the adjacency matrix of $G$, and view $A_G$ as an element of the vector space $\F_2^{\binom{n}{2}}$.
For two graphs $G$, $H$, let $G \Delta H$ be the symmetric difference of the two graphs, i.e. $A_{G \Delta H} = A_G - A_H$.
For a subset $S \subseteq [n]$, we use $G[S]$ to denote the subgraph of $G$ induced by the vertex set $S$. 
If $A$ is a $n\times n$ matrix and $S, T \subset [n]$, let $A_{S, T}$ be the sub-matrix indexed by $S$ on the rows and $T$ on the columns.
For $x \in [0,1]$, we use $h_2(x) = - x \log_2 x - (1-x)\log_2(1-x)$ to denote the binary entropy function.

\begin{definition}[Graph distance and relative graph distance]\label{def:graph-dist}~
    \begin{itemize}
        \item The {\em graph distance} between two graphs $G$ and $H$, denoted by $\dgraph(G,H)$, is the smallest $d \in \N$ such that there is a set $S\subseteq [n]$, $|S| = d$, and $G[[n] \setminus S] = H[[n] \setminus S]$. 
        \item The {\em relative graph distance}, or simply relative distance, between $G$ and $H$ is denoted by $\delta_{\text{graph}}(G,H)$, and is the quantity $\frac{\dgraph(G,H)}{n}$.
    \end{itemize}
\end{definition}

In the above definition, we require that the graphs $G[\overline{S}]$ and $H[\overline{S}]$ be identical and not just isomorphic. 
Lemma~\ref{prop:dgraph-characterizations} describes several equivalent characterizations of graph distance.

\begin{proposition}[Alternate characterizations of $\dgraph$]\label{prop:dgraph-characterizations}
Suppose $G$ and $H$ are graphs on the same vertex set. Then 
    \begin{enumerate}[1.]
        \item $\dgraph(G, H)$ is the minimum vertex cover size of $G \Delta H$.
        \item $\dgraph(G, H)$ is the minimum number of vertices whose neighborhoods you need to edit to transform $G$ into $H$
        \item $\dgraph(G, H)$ is the minimum number of vertices whose neighborhoods you need to edit to transform $G \Delta H$ into the empty graph.
    \end{enumerate}
\end{proposition}

Note that $\dgraph$ is a metric (see \cite{yohananovCodesGraphErasures2019}, Lemma 5).

\begin{definition}[Graph code]
We say that a set $C \subseteq 2^{\binom{[n]}{2}}$ is a {\em graph code} on $[n]$ with distance $d$ if for every pair of graph $G,H\in C$, we have that $\dgraph(G,H) \geq d$.
\begin{itemize}
\item The {\em rate} of $C$, denoted by $R_C$, is the quantity $\frac{\log_2(|C|)}{\binom{n}{2}} \geq\frac{2\log_2(|C|)}{n^2}$. 
\item The {\em distance} (resp. {\em relative distance}) of $C$, denoted by $d_C$ (resp. $\delta_C$), is the quantity $\min_{\substack{G,H \in C\\G \neq H}}\dgraph(G,H)$ (resp. $\min_{\substack{G,H \in C\\G \neq H}}\delta_{\text{graph}}(G,H)$).
\end{itemize}
\end{definition}

\paragraph{Upper bound.}
As noted in \cite{yohananovCodesGraphErasures2019}, the Singleton bound can be used to obtain an upper bound on the rate of a graph code. We include a proof here for completeness.

\begin{proposition}\label{prop:ub}
Any graph code with relative distance $\delta$ has dimension at most $\binom{n(1 - \delta) + 1}{2} $.
\end{proposition}

\begin{proof}
    Consider any graph code $C$ of relative distance $\delta$. Let $A \subset [n]$ be any subset of at most $\delta n - 1$ vertices. For any two distinct $G_1,G_2\in C$, we have that the graphs induced on the vertices outside $A$, $G_1[[n] \setminus A]$ and $G_2[[n] \setminus A]$, are different. Indeed, since otherwise, $A$ is a vertex cover of $G_1 \Delta G_2$, contradicting the relative distance assumption. So, we have that 
    $$|C| \leq 2^{\binom{n(1-\delta) + 1}{2}}.$$
\end{proof}

Expressed in terms of rate and distance, \Cref{prop:ub} implies that for constant relative distance $\delta$, $R \leq (1 - \delta)^2 + O(1/n)$.

\section{Existence of optimal graph codes}\label{sec:combinatorics}

As with other objects in the theory of error-correcting codes, the first question we seek to answer relates to the optimal rate-distance tradeoff. 

In contrast to the Hamming world, we find that random linear codes meet the Singleton bound in the graph distance.

\begin{proposition}\label{prop:lb}
    Let $\delta \in (0,1)$. Then, there exists a linear graph code with distance greater than $\delta$ and dimension at least 
    \[
    \max\left\{\binom{n(1 - \delta)}{2} - h_2(\delta)n - 2,0\right\}.
    \]
\end{proposition}

\newcommand{\va}{\vec{\alpha}}
\newcommand{\vb}{\vec{\beta}}
\newcommand{\vz}{\vec{0}}

\begin{proof}
    We only consider the case when $\binom{n(1 - \delta)}{2} - h_2(\delta)n - 2>0$, and prove this by a probabilistic construction. 
    Let $\mathbf{G}_{n,1/2}$, be the Erd\"{o}s-R\'{e}nyi random graph distribution where the vertices are $[n]$, and each of the $\binom{n}{2}$ possible edges are selected independently with probability 1/2.
    Let $k = \binom{n(1 - \delta)}{2} - h_2(\delta)n - 2$, and let $G_1,...,G_k$ be graphs chosen independently from $\mathbf{G}_{n, 1/2}$.
    Consider the $\mathbb{F}_2$-linear space $C = \mathrm{span} \{A_{G_1},\ldots,A_{G_k}\}$.
    We wish to show that $C$ has distance at least $\delta n$ with high probability.

    For $\va \in \F_2^k$, let $H_{\va}$ be the graph with the adjacency matrix $\sum_{i=1}^k \alpha_i A_{G_i}$.
    Recalling the definition of graph distance, we need that for any distinct $\va, \vb \in \F_2^k$, $H_{\va} \Delta H_{\vb}$ must have minimum vertex cover size at greater than $\delta n$.
    Since $H_{\va} \Delta H_{\vb} = H_{\va - \vb}$, it suffices to show that for any non-zero $\va \in \F_2^k$, that $H_{\va}$ has no vertex cover of size $\delta n$.

    For any $\vec{\alpha} \in \mathbb{F}_2^{k} \setminus \{\vec{0}\}$, $H_{\vec{\alpha}}$ has the same law as $\mathbf{G}_{n,\frac{1}{2}}$. 
    Let $B_{\va}$ be the event that $H_{\vec{\alpha}}$ has a vertex cover of size $\delta n$. 
    We have
    \begin{align*}
        \Pr(B_{\va}) &= \Pr\left(\exists S \in \binom{[n]}{\delta n}: G[[n] \setminus S] = \text{ is the empty graph} \right)\\
                      &\leq \binom{n}{\delta n}\cdot 2^{-\binom{n(1-\delta)}{2}} \\
                      &\leq 2^{-\binom{n(1 - \delta)}{2} + h_2(\delta) n},
    \end{align*}
    where the first inequality uses the union bound over subsets of size $\delta n$, and the fact that there are $\binom{n(1 - \delta)}{2}$ edges outside of a vertex cover that all have to be unselected.
    Then, union bounding over the choices of $\vec{\alpha}$, we get
    \begin{align*}
        \Pr\left(\bigcup_{\alpha \in \F_2^k \setminus \{\vz\}}  B_{\va}\right) \leq 2^{k-\binom{n(1 - \delta)}{2} + h_2(\delta) n} \leq 1/4.
    \end{align*}

    Therefore, with probability at least $3/4$, there does not exist $\alpha \in \F_2^k$, such that $\alpha \neq \vz$, and $H_\alpha$ has a vertex cover of size $\delta n$, and hence $C$ is a code with relative distance greater than $\delta$.

    Finally, note that if $H_\alpha$ does not have a vertex cover of size $\delta n$, then $H_\alpha$ is not the empty graph.
    Thus, the fact that there does not exist $\alpha \in \F_2^k$, $\alpha \neq \vec{0}$ such that $H_\alpha$ has a vertex cover of size $\delta n$ implies that $G_1,..., G_k$ are linearly independent.
    Hence, the dimension of $C$ is $k = \binom{n(1 - \delta)}{2} - h_2(\delta)n - 2$, as required.
\end{proof}

As a result, we have the following corollaries.

\begin{corollary}
    For any constant $\delta \in (0,1)$, there exist optimal linear graph codes with relative distance at least $\delta$.
\end{corollary}

\begin{corollary}
    For any constant $c > 2$, there exists a linear graph code with dimension at least $\Omega(\log^2 n)$ and relative distance at least $\delta = 1 - c\cdot \frac{\log n}{n}$.
\end{corollary}

\section{Explicit graph codes for high distance: Concatenated Codes}\label{sec:concatenation}

To get explicit codes of distance $\delta > 1/2$, we start with Construction 3 of \cite{yohananovCodesGraphErasures2019}.
The construction utilizes the tensor product code introduced by \cite{wolfCodesDerivableTensor1965}, where elements of the code are matrices where all rows and columns are codewords over some base Hamming code. 
The elements of the tensor code are then the adjacency matrices of graphs. 
Since we consider undirected graphs and do not allow self-loops, we take the subcode of the tensor code containing only symmetric matrices with zeros on the diagonal.

\begin{definition}[Symmetric Tensor Code with Zeros on the Diagonal]\label{def:stczd}
	Let $C$ be a code over $\F_q$. The symmetric tensor code with zeros on the diagonal built on $C$ denoted $\STCZD(C)$ is the set of matrices $A$ over $\F_q^{n\times n}$ such that (1) $A$ is symmetric, (2) the rows and columns of $A$ are codewords of $C$, and (3) the entries on the diagonal are all 0.
\end{definition}

Properties of elements of Tensor Product Codes that are symmetric and zero-diagonal were also previously studied in the context of constructing a gap-preserving reduction from SAT to the Minimum Distance of a Code problem by Austrin and Khot \cite{austrinSimpleDeterministicReduction2014}.

We will also define another notion of distance that will be useful later.

\begin{definition}[Directed graph distance]
    Let $A$, and $B$ be $n\times n$ matrices over some field. Define the \emph{directed graph distance} denoted $\dab(G, H)$ to be the minimum $d$ such that there exists sets $S, T \subset [n]$ of size $d$ where $(A - B)_{\overline{S}, \overline{T}} = \mathbf{0}$.
\end{definition}

For weighted, directed graphs, $G$, and $H$, abbreviate $\dab(A_G, A_H) = \dab(G, H)$. To better distinguish between $\dab$ and $\dgraph$, we sometimes refer to $\dgraph$ as the \emph{undirected} graph distance.

When $G$ and $H$ are weighted directed graphs, $\dab(A_G, A_H)$ can be viewed as the minimum $d$ such that you can go from $G$ to $H$ by editing the incoming edges of $d$ vertices and the outgoing edges of $d$ vertices.
The main difference between directed and undirected graph distance is that directed graph distance allows the subset of rows and the subset of columns to be edited to be different. 
Insisting that $S = T$ in the definition for directed graph distance recovers the undirected graph distance. 
From this, it easily follows that if $G$ and $H$ are undirected graphs, then $\dab(G, H) \leq \dgraph(G, H)$. 
Thus, to find codes with high graph distance, it suffices to find codes with large directed graph distance, where all the elements are adjacency matrices of undirected, unweighted graphs (i.e., 0/1 matrices that are symmetric and zero diagonal). 
Note that when discussing rate directed graph codes $C$, we are referring to the quantity $\log_q(|C|)/n^2$ instead of $\log_q(|C|)/\binom{n}{2}$.

In the next lemma, we show several properties of $\STCZD(C)$. 
Most importantly, the Hamming distance of $C$ translates to the directed graph distance of $\STCZD(C)$.

\begin{lemma}\label{lem:symmetric-tensor-properties}
    Let $C$ be a linear $[n, k, d]_q$-code, then $\STCZD(C) \subset \F_q^{n \times n}$ is linear, has dimension at least $\binom{k+1}{2}-n$, and has directed graph distance $d$.
\end{lemma}

\begin{proof}
	Let $C$ be a linear $[n, k, d]_q$-code, and let $C'=\STCZD(C)$. $C'$ is linear because $C$ is linear, and the sum of symmetric matrices is symmetric.

    WLOG, we assume that $C$ is systematic, i.e., it has $k \times n$ generator matrix $G = [I|A]$, where $I$ is the $k \times k$ identity and $A$ is a $k \times (n-k)$ matrix. Then, for every $X \in \F_q^{k \times k}$, the following has rows and columns belonging to $C$
	\[
		G^T X G =
		\begin{bmatrix}
			X    & XA    \\
			A^TX & A^TXA \\
		\end{bmatrix}.
	\]
	Furthermore, $G^TXG$ is symmetric and has zeros on the diagonal iff $X$ is symmetric, $X$ has zeros on the diagonal, and $A^TXA$ has zeros on the diagonal. 
    This imposes $\binom{k+1}{2} + (n-k)$ linear constraints on the entries of $X$. 
    Thus, the subspace of $X$ for which $G^TXG \in C'$ has dimension at least $k^2 - \binom{k+1}{2} - (n-k) = \binom{k+1}{2}-n$. 

	Since $C'$ is linear, to show the distance property, it suffices to show that $\dab(A, \mathbf{0}) \geq d$ for every non-zero $A \in C'$. 
    Let $A \in C'$ be a non-zero element of $C'$, we'll show that for any $S, T \subset [n]$, with $|S| < d$, and $|T|<d$, $A_{\overline{S}, \overline{T}} \neq \mathbf{0}$.

	Since $A$ is non-zero, there is some non-zero entry $A_{ij}$. 
    Since the rows are elements of a linear code of distance $d$, the Hamming weight of the $i$th row is at least $d$. 
    Since $|T| < d$, there is some $j' \notin T$ such that $A_{ij'}$ is non-zero. 
    Then, the $j'$th column is also a non-zero codeword of $C$, so it has a Hamming weight of at least $d$. 
    Since $|S| < d$, there is some $i' \notin S$ such that $A_{i'j'}$ is non-zero. 
    Thus, $A_{\overline{S}, \overline{T}} \neq \mathbf{0}$.
\end{proof}

\begin{remark}
A simple calculation shows that if $C$ has constant rate, $R$, then $\STCZD(C)$ has rate $R^2/2-o(1)$ as a directed graph code.
\end{remark}

Given this lemma (and using the fact that $\dgraph \geq \dab$), for any \textbf{binary} code $C \subset \F_2^{n}$ with rate $R$ and relative distance $\delta$, $\STCZD(C)$ is a (undirected) graph code with rate $R^2 - o(1)$, and relative distance $\delta$. 
Thus, if $C$ has rate distance tradeoff $R = f(\delta)$, then $\STCZD(C)$ has rate distance tradeoff $R = f(\delta)^2 - o(1)$.  
Immediately, we get that taking the $\STCZD$ of any asymptotically good binary code yields an asymptotically good graph code.

There are two problems with this construction. 
Firstly, these codes may not be strongly explicit. 
Secondly, the Plotkin bound \cite{plotkinBinaryCodesSpecified1960} implies that any binary code with distance $>1/2$ has vanishing rate. 
So this falls short of our goal of obtaining strongly explicit, asymptotically good codes with $\delta > 1/2$. 

We will address the first problem by showing that if the base code is a Reed Solomon code \cite{reedPolynomialCodesCertain1960}, then there is a large subcode that is strongly explicit.

\begin{code}[Reed Solomon Code $\mathrm{RS}(n, R, q)$ \cite{reedPolynomialCodesCertain1960}]\label{code:RS}
The Reed Solomon Code with parameters $n$, $R$, and $q$, where $q \geq n$, is a code over $\F_q^n$ with rate $R$ and distance $1 - R$.
\end{code}

\begin{lemma}\label{thm:stczd-rs-se}
    Let $C \in \mathrm{RS}(n,R,q)$ where $Rn = k-1$. 
    Then, there exists a \emph{strongly explicit} subcode $S \subset \STCZD(C)$ such that the dimension of $S$ is at least $\binom{k-1}{2}$.
\end{lemma}

\begin{proof}
    Essentially, we will evaluate symmetric polynomials that are a multiple of $(X-Y)^2$ on a $n \times n$ grid. 

    Suppose $h(X, Y)$ is a symmetric polynomial of individual degree at most $k - 3$, and let $M$ be the evaluations of $f(X, Y) = (X - Y)^2h(X, Y)$ on a $n \times n$ grid. 
    $M$ is symmetric and has zeros on the diagonal. 
    Furthermore, for a fixed value, $y$, $f(X, y)$ is a univariate polynomial in $X$ of degree at most $k-1$, and hence the column indexed by $y$ is an element of a Reed Solomon code of dimension $k$, and block length $n$. 
    Similarly, the rows are also elements of the same code. Thus $M \in \STCZD$.

    Let $S$ be the space of bivariate symmetric polynomials of degree at most $k-3$. 
    For $a, b \in \N$, define polynomials $p_{a,b}(X, Y) = X^aY^b + X^bY^a$. 
    Notice that $p_{a, b}$ is symmetric, and furthermore the set 
    $$\{p_{a, b}: 0 \leq a < b \leq k - 3\} \cup \{X^iY^i: i \in \{0,1,...,k - 3\}\},$$
    is linearly independent. 
    Thus $\dim(S) = \binom{k - 2}{2} + k -2  = \binom{k - 1}{2}$, as desired.
    
\end{proof}

To extend this construction to the setting of $\delta > 1/2$, we use the concatenation paradigm from standard error-correcting code theory, initially introduced by Forney \cite{forney1965concatenated}. 

We will start with a code over a large alphabet and then concatenate with an inner code, which will be an optimal directed graph code. 

\begin{lemma}\label{lem:inner}
For any $\epsilon > 0$, and sufficiently large $n$, for any $k < \epsilon^2n^2 - 2n$, there exists a linear directed graph code over $\F_2$ of dimension $k$ and distance at least $(1 - \epsilon)n$.
\end{lemma}

The proof is standard and similar to that of Proposition~\ref{prop:lb}. 
So we will omit it.

\begin{code}[Optimal Directed Graph Code $\mathrm{Opt}(\epsilon, n, k)$]\label{code:opt}
    Require $k < \epsilon^2n^2  - 2n$. 
    Refer to a code with the properties in \Cref{lem:inner} as $\mathrm{Opt}(\epsilon, n, k)$. 
\end{code}

\subsection{Symmetric concatenation}

Since our inner code is not guaranteed to be symmetric, simply replacing each field element in the outer code with its encoding might result in an asymmetric matrix. 
To remedy this, we transpose the encoding for entries below the diagonal. 
This is made formal below.

\begin{definition}[Symmetric Concatenation]\label{def:concat}
    Let $q, Q$ be prime powers, and $n, N$ be positive integers.  
    Let $C_{out} \subset \F_{Q}^{N \times N}$ and  $C_{in} \subset \F_{q}^{n \times n}$ such that $|C_{in}| = Q$. 
    Define $C_{in}\circ C_{out} \subset \F_{q}^{nN\times nN}$ to be the code obtained by taking codewords of $C_{out}$ and replacing each symbol of the outer alphabet with by their encodings under $C_{in}$ if they lie above or on the diagonal, and with the transpose of their encodings if they lie below the diagonal. 
\end{definition}

\begin{figure}[ht]
    \centering
    \includegraphics[width=\textwidth]{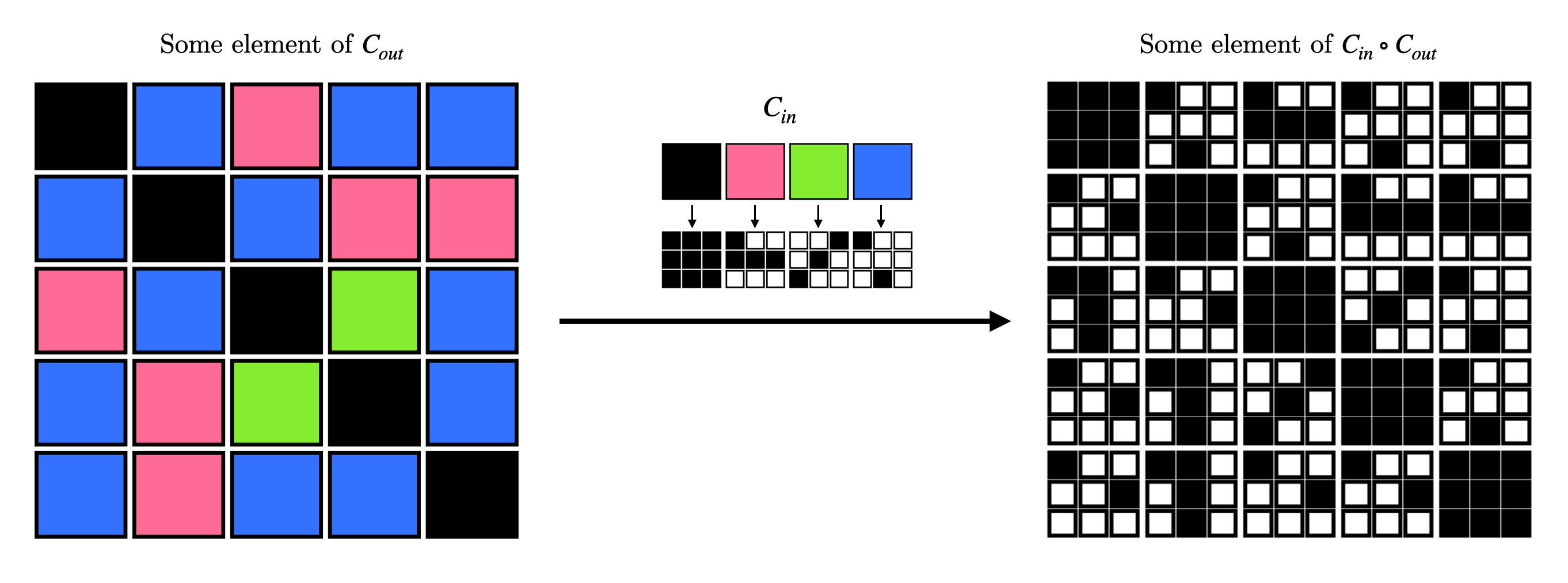}   
    \caption{
        Example of symmetric concatenation. 
        An outer codeword is shown on the left, with field elements represented as different colors. 
        The concatenation with the inner code is shown to the right. 
        Black squares represent $0$, and white squares represent $1$s.
    }
    \label{fig:sym-concat}
\end{figure}

\Cref{fig:sym-concat} visualizes an example of symmetric concatenation. 
We now show that distance and dimension concatenate exactly like it does for standard error-correcting codes.

\begin{lemma}\label{lem:concat}
    Suppose $C_{in}$ and $C_{out}$ are linear codes as in the previous definition with directed graph distance $d$ and $D$, respectively. 
    Let $k$ be the dimension of $C_{in}$, and $K$ be the dimension of $C_{out}$. 
    Note $|C_{in}| = q^{k} = Q$. 
    Then $C_{in} \circ C_{out}$ is linear and has distance at least $dD$, and dimension $kK$.
\end{lemma}

\begin{proof}
Let $C = C_{in} \circ C_{out}$. 
First note that $C_{in} \circ C_{out}$ can be made linear by using a $\F_{q}$-linear map from $\F_{q^k}$ to $\F_q^{k}$ before encoding with the inner code.

Consider a non-zero outer codeword $O$, and let $A$ be the codeword after concatenation. 
Let $S, T \subset [nN]$ be of size less than $dD$. 
We'll show that $A_{\overline{S}, \overline{T}} \neq \mathbf{0}$. 
Partition $A$ into $N \times N$ blocks, where the $(I, J)$'th block for $I, J \in [N]$, is the $n \times n$ matrix encoding the symbol at $O_{IJ}$. 
Identify the indices $[nN]$ with $[N] \times [n]$ where the tuple $(I, i)$ corresponds to the $i$'th index in the $I$'th block.

For $I \in [N]$, let $S_I = \{i \in [n]: (I, i) \in S\}$ be the set rows in $S$ in the $I$'th block. 
Define $T_J$ similarly. 
Let $S_{\geq} = \left\{ I \in [N]: |S_I| \geq d \right\}$, be the set of blocks in which there are at least $d$ elements in $S$, and similarly define $T_{\geq}$, $S_{<}$, and $T_<$.
 
Since $\sum_{I \in [N]}|S_I| < dD$, $\sum_{J \in [N]}|T_J| < dD$, we have $|S_{\geq}| < D$, and $|T_{\geq}| < D$. 
Since the outer code has directed distance $D$, $O_{S_<, T_<} \neq \mathbf{0}$, so there exists $I \in S_<$, and $J \in T_<$ such that $O_{I,J}$ is non-zero. 
So, the $(I, J)$'th block of $A$ is a non-zero codeword or the transpose of a non-zero codeword of $C_{in}$. 
Let us call it $X$, and suppose that $X \in C_{in}$.

Since $|S_I| < d$, and $|T_J| < d$, and the inner code has distance at least $d$, we have that $X_{\overline{S_I}, \overline{T_J}} \neq \mathbf{0}$. 
To finish the proof, note that $\dab(X, \mathbf{0})=\dab(X^T, \mathbf{0})$ by switching the roles of $S$ and $T$ in the definition of directed graph distance. 

For the claim about dimension, note that the number of codewords in $C$ is the number of codewords in $C_{out}$, which is $Q^K$. 
The dimension of $C$ is then $K\log_q(Q) =Kk$. 

\end{proof}

Additionally, it is clear from the definition of symmetric concatenation that if $C_{out}$ is symmetric and zero-diagonal, so is $C_{in}\circ C_{out}$. 

\begin{remark}
    \Cref{lem:concat} also holds for the standard definition of concatenation (without transposing blocks below the diagonal). 
    However, we will not need this fact.
\end{remark}

\subsection{Concatenated graph codes}

We can instantiate the concatenated code using Reed Solomon codes. 

\begin{code}[Concatenated Code $\CRS(\epsilon, n, k, N, \rho)$]\label{code:concat-RS}
Let $Q = 2^k$ be the size of the alphabet of the outer code.  
Let $\epsilon, \rho \in (0,1)$, and $n, k, N$, to be integers satisfying $k < \epsilon^2n^2 - 2n$, and $N \leq Q$. 
Then, 
    $$
    \CRS =   \mathrm{Opt}(\epsilon, n, k) \circ \STCZD(\mathrm{RS}(N, \rho, Q)).
    $$
\end{code}

The following theorem follows directly from \Cref{lem:symmetric-tensor-properties,lem:concat}. 
As a reminder, here we are considering the rate of the codes as a (undirected) graph code.

\begin{theorem}\label{thm:concat-RS properties}
Let $\epsilon, n, k, N, \rho$ be parameters satisfying the requirements listed in \Cref{code:concat-RS},
then $\CRS(\epsilon, n, k, N, \rho)$ is a graph code with rate $\epsilon^2\rho^2 - o(1)$, and relative distance $(1-\epsilon)(1-\rho)$. 
\end{theorem}

Note that using this construction, we can get asymptotically good codes for any constant rate and distance - including for distances $> 1/2$, which was not obtained in any of the previous constructions. 
We get $ R = (1 - \sqrt{\delta})^4 - o(1)$ by setting $\epsilon = \rho$.

One drawback of this construction is that it is not strongly explicit or even explicit. 
The outer code can be made strongly explicit using \Cref{thm:stczd-rs-se}; however, the inner code was an optimal code, which we obtained by a randomized construction. 
The complexity of searching for such a code by brute force is too large. 
In particular, the optimal code has dimension $\epsilon^2n^2$ and block length $n^2$. 
Since we need the size of the code to be equal to the size of the outer alphabet, we have $N = 2^{\epsilon ^2n^2}$, so $n = \sqrt{\log(N)}/\epsilon$. 
Then, there are at least $2^{\epsilon^2n^4} = 2^{\log(N)^2 / \epsilon^2}$ generator matrices to search over. 
Thus, we cannot find such a code efficiently. 

To address this, we reduce the search space by concatenating \emph{multiple times}. 
The resulting code will have a slightly worse distance/rate tradeoff but will still be asymptotically good for any constant distance or rate. 

We note that $\CRS$ can also be made strongly explicit using a Justensen-like construction. 
However, although this code is again asymptotically good, it has distance bounded away from 1. 
We present this construction in the Appendix. 

\subsection{Multiple concatenation}

While concatenating twice suffices to obtain an explicit code, it is not clear that the obtained code is strongly explicit. 
This may be addressed by concatenating three times at the cost of slightly weaker parameters. 
Here, we will also use the tensor product code as a building block. 
For any linear code $C \subset \F_q^n$, let $\mathrm{TC}(C)$ be the tensor product code of $C$. 
As a reminder, $\mathrm{TC}(C)$ is the code consisting of matrices $A \in \F_q^{n \times n}$ such that each row and each column of $A$ are elements of $C$. 

\begin{remark}
Suppose $C$ is a linear code with distance $d$ and rate $R$.  
Then, it follows from the proof of \Cref{lem:symmetric-tensor-properties} that $\mathrm{TC}(C)$ has directed graph distance at least $d$. 
It is also well known that $\mathrm{TC}(C)$ has rate $R^2$.
\end{remark}

Below, we present the analysis for triple concatenation.
\newcommand{\CTrip}{\mathrm{C_{Trip}}}
\begin{code}[Triple Concatenation $\CTrip(\rho, N)$]\label{code:triple-concatenation}
    For $\rho \in (0,1)$ and an integer $N$, let $C$ be the subcode of $\STCZD(\mathrm{RS}(N, \rho, N))$ in \Cref{thm:stczd-rs-se}. Then 
    $$\CTrip = \mathrm{Opt}(N_3, \rho) \circ \mathrm{TC}(\mathrm{RS}(N_2, \rho, N_2)) \circ \mathrm{TC}(\mathrm{RS}(N_1, \rho, N_1)) \circ C,$$
    where $N_1, N_2$ and $N_3$ are picked to make the concatenation work, i.e., $|\mathrm{Opt}(N_3, \rho)| \geq N_2$, and so on.
\end{code}

Notice that only the outer-most code needs to be symmetric and have zero diagonal since we use the symmetric concatenation operation (entries below the diagonal will be transposed). 
Thus, using the Tensor Product Code for the two middle codes (instead of STCZD) is sufficient and saves a factor of 2 (each time) on the rate.

\begin{theorem}\label{thm:triple-concatenation}
    Let $N$ be a positive integer and $\rho \in (0,1)$. Then $\CTrip(\rho, N)$ has distance at least $(1 - \rho)^4$, and rate $\rho^8$. Furthermore, $\CTrip(\rho, N)$ is strongly explicit.
\end{theorem}

\begin{proof}
    The claims about rate and distance follow directly from \Cref{lem:concat}.

    We'll now show that this code is strongly explicit. 
    The outermost code $C$ is strongly explicit, and the two codes in the middle built on Reed Solomon codes are also strongly explicit. 
    The idea is that the concatenation steps in the middle allow us to shrink the alphabet size from $N$ to (less than) $\log(\log(N))$. 
    Searching for optimal codes of this size can be done easily by brute force.

    The dimension of $\mathrm{TC}(\mathrm{RS}(N_1, \rho, N_1))$ is $(\rho N_1)^2$, so the number of codewords is $N_1^{(\rho_1N_1)^2}$, and for the concatenation to work, we need this to be at least $N$. 
    That is, we need $(\rho_1N_1)^2\log(N_1) \geq \log(N)$, which we can get easily by setting $N_1 = O(\log(N))$. 
    For the same reason, we can take $N_2=O(\log\log(N))$.

%
%
%

    This is now small enough to do a brute-force search for an optimal code. 
    The inner-most code has dimension $\rho^2N_3^2$, so we need $2^{\rho^2N_3^2} = N_2$, or $N_3 = O(\sqrt{\log(N_2)})$. 
    There are then $\rho^2N_3^4$ possible generator matrices to search over. 
    So the total number of codes we will need to search over is at most $2^{\rho^2n^4} = 2^{O(\log\log\log(N)^2)} = 2^{o(\log\log(N))} = O(\log(N))$.
\end{proof}

The tradeoff for this code is then 
$$R = (1 - \delta^{1/4})^8.$$

Thus, we get \emph{strongly explicit} asymptotically good codes for any constant distance or rate.

Concatenating twice would suffice if we just wanted explicit codes (instead of strongly explicit). 
In particular, the search space for the inner-most code has size 
$$
2^{O(\log\log(N)^2)} = 2^{o(\log(N))},
$$
which is smaller than any polynomial in $N$ but not polylogarithmic. 
The corresponding tradeoff for the double concatenated code is $R = (1 - \delta^{1/3})^6$.

\section{Explicit graph codes with very high distance: dual-BCH Codes}\label{sec:weil}

In this section, we give explicit constructions of graph codes
for the setting of very high distance ($\delta = 1 - o(1)$). 
As noted earlier, when the complete graph and the empty graph are part of the code, this is a generalization of the problem of constructing explicit Ramsey graphs (i.e., graphs with no large clique or independent set), which corresponds to graph codes of size at least $3$.

Our main result here is an explicit construction of a graph code with distance $1- \frac{n^{\epsilon}}{n^{1/2}}$ and dimension $\Omega(n^{\epsilon}\log n)$, for all $\epsilon \in [0, 1/2)$.

\begin{theorem}
For all $d$, there is a strongly explicit construction construction of a code with 
dimension $\Omega(d \log n)$ and distance $n - O(d \sqrt{n})$.
\end{theorem}

In analogy with the situation for Hamming-distance codes, these are the dual-BCH codes of the graph-distance world.

\subsection{Warmup: a graph code with dimension \texorpdfstring{$\Omega(\log n)$}{Omega(log(n))}}

As a warmup, we first construct code with distance $1 - \frac{1}{n^\epsilon}$ with growing dimension. 

Let $n = 2^t$.
Let $\Tr : \F_{2^t} \to \F_{2}$ denote the finite field trace map.
Concretely, it is given by:
$$\Tr(x) = x + x^2 + x^4 + \ldots + x^{2^i} + \ldots + x^{2^{t-1}}$$
For each $\alpha \in \F_{2^t}$,
consider the matrix $M_\alpha \in \F_2^{n \times n}$, where the rows and columns of $M_\alpha$ are indexed by elements of $\F_{2^t}$, given by:
$$ \left(M_{\alpha}\right)_{x,y} = \Tr(\alpha \cdot (x+y)^3).$$

Note that each $M_\alpha$ is symmetric. Consider the code
\begin{code}\label{code:highdistwarmup}
For $n$ of the form $2^t$, let us define the family of codes
\[
C_{\mathrm{Warmup}} = \{ M_\alpha \mid \alpha \in \F_{2^t} \}.
\]
\end{code}
We have that $C_{\mathrm{Warmup}}$ is a linear code of dimension $t = \log_2 n$.

\begin{theorem}\label{thm:DualBCH}
The distance of $C_{\mathrm{Warmup}}$ is at least $1 - O(n^{-1/2})$.

\end{theorem}
\begin{proof}
Fix any $\alpha \in \mathbb{F}_{2^t}$. Let $S \subseteq \F_{2^t}$ be an arbitrary subset of vertices.
It suffices to show that if $S$ is bigger than $\Omega(n^{1/2}) = \Omega(2^{t/2})$, then there exist some $x,y \in S$ such that 
$$\Tr( \alpha \cdot (x+y)^3 ) = 1.$$

Suppose not. Then we have:
$$ \sum_{x, y \in S} (-1)^{\Tr( \alpha (x+y)^3 )} = |S|^2.$$
By Cauchy-Schwarz, we get:
\begin{align*}
 |S|^4 &= \left(\sum_{x \in S} \sum_{y \in S} (-1)^{\Tr(\alpha (x+y)^3)}\right)^2 \\
 &\leq \left(\sum_{x \in S} \left(\sum_{y \in S} (-1)^{\Tr(\alpha (x+y)^3)}\right)^2 \right)   \cdot |S|\\
 &\leq \left(\sum_{x \in \F_{2^t}} \left(\sum_{y \in S} (-1)^{\Tr(\alpha (x+y)^3)}\right)^2 \right)   \cdot |S|\\
 &= \left(\sum_{x\in \F_{2^t}} \sum_{y_1, y_2 \in S} (-1)^{\Tr(\alpha ((x+y_1)^3 + (x+y_2)^3))}   \right) \cdot |S|\\
 &\leq \left(\sum_{y_1, y_2}  \left| \sum_{x \in \F_{2^t}} (-1)^{\Tr(\alpha ((x+y_1)^3 + (x+y_2)^3))} \right|  \right) \cdot |S|.
\end{align*}

For $y_1, y_2 \in \F_{2^t}$, let $P_{y_1, y_2}(X)$ be the polynomial
\begin{align*}
P_{y_1, y_2} (X ) &=  \alpha \cdot \left( (X+y_1)^3 + (X+ y_2)^3 \right)\\
&=\alpha \cdot \left( (y_1 + y_2) X^2 + (y_1^2 + y_2^2)X + (y_1^3 + y_2^3)\right).
\end{align*}

The key observation is that for most $(y_1, y_2) \in S^2$, the trace of the polynomial $P_{y_1,y_2}(X)$ is a nonconstant $\F_2$-linear function, and thus the inner sum:
$$ \sum_{x \in \F_{2^t}} (-1)^{\Tr(P_{y_1,y_2}(x))}$$
equals $0$.

\begin{lemma}
 If $P(X) = a X^2 + bX + c \in \F_{2^t}[X]$, then 
 $$ \Tr \circ P : \F_{2^t} \to \F_2$$
 is a nonconstant $\F_2$-linear function unless 
 $a = b^2$.
\end{lemma}
The proof is standard, and we omit it.

By the lemma, we get that there are at most $4|S|$ choices of $(y_1, y_2)$ such that the inner sum is non-zero (namely those $(y_1, y_2) \in S^2$ for which $\alpha (y_1 + y_2) = \left(\alpha (y_1^2 + y_2^2) \right)^2$, which are few in number by the Schwartz-Zippel lemma).

Thus we get:
$$ |S|^4 \leq 4|S|^2 \cdot 2^t,$$
from which we get  $|S| \leq  O(2^{t/2})$, as desired.
\end{proof}

\subsection{Larger dimension}

We now see how to get graph codes of distance $1 - \frac{1}{n^\epsilon}$ with $\epsilon < \frac{1}{2}$ and larger rate.

For a polynomial $f(X) \in \F_{2^t}[X]$, let $M_f$ be $n \times n$ matrix with rows and columns indexed by $\F_{2^t}$ for which:
$$ \left( M_f\right)_{x,y} = \Tr( f(x+y) ).$$

Let $W_d$ be the $\F_{2^t}$-linear space of all polynomials $f(X)$ of the form:
$$ f(X) = \sum_{3 \leq 2i+1 \leq  d} \alpha_i X^{2i+1},$$
where the $\alpha_i \in \F_{2^t}$.

Then, let us define our construction.
\begin{code}\label{code:dualBCH}
For $n$ of the form $2^t$ and $d \leq n^{1/2}$, let us define the family of codes
\[
C_{\mathrm{DualBCH}}(n,d) = \{ M_f : f \in W_d \}.
\]
\end{code}

\begin{theorem}
We have that $C_{\mathrm{DualBCH}}(n,d)$ is a linear graph code of distance $1 - O\left(dn^{-1/2}\right)$ and dimension $\Omega(dt) = \Omega(d\log n)$.

\end{theorem}
\begin{proof}
 The proof is very similar to the proof of Theorem~\ref{thm:DualBCH}. 
 Consider any $M_f \in C_{\mathrm{DualBCH}}(n,d)$ with $f \neq 0$. 
 It suffices to show that the independence number\footnote{An essentially identical proof shows that the clique number also has the same bound. 
 The only change is to replace the LHS of~\eqref{eqS2} by $-(|S|^2 - |S|)$, and this sign change does not affect anything later because we immediately apply Cauchy-Schwarz to get~\eqref{eq:dbcheq}. 
 
 This justifies our referring to this code as a ``code of Ramsey graphs".} of $M_f$ is $O(dn^{1/2})$.
 
 Assume that $S \subseteq \F_{2^t}$ is an independent set in $M_f$.
 Then 
 \begin{align}
     \label{eqS2}
     |S|^2 = \sum_{x,y \in S} (-1)^{\Tr(f(x+y))}.
 \end{align}
 
 As in the proof of Theorem~\ref{thm:DualBCH}, by the Cauchy-Schwartz inequality and some simple manipulations, we get:
 \begin{align}
 \label{eq:dbcheq}
  |S|^4 &\leq \left( \sum_{y_1, y_2 \in S} \left| \sum_{x \in \F_{2^t}} (-1)^{\Tr(P_{y_1,y_2}(x))} \right|   \right) \cdot |S|,
 \end{align}
 where:
 $$ P_{y_1, y_2}(X) = f(X+y_1) - f(X+y_2).$$

 At this point, we need an upper bound in the inner sum:
 $$ U_{y_1, y_2} = \left| \sum_{x \in \F_{2^t}} (-1)^{\Tr(P_{y_1, y_2}(x))}    \right|.$$

To get this, we will invoke the deep and powerful Weil bound:
 \begin{theorem} [\cite{Schmidt1976EquationsOF}, Chapter II, Theorem 2E]
 \label{thm:weil}
  Suppose $P(X) \in \F_{2^t}(X)$ is a nonzero polynomial of odd degree
  with degree at most $d$. Then:
  $$ \left| \sum_{x \in \F_{2^t}} (-1)^{\Tr(P(x))}    \right| \leq O(d 2^{t/2}).$$
 \end{theorem}

 We will use this to show that all but a few pairs $(y_1, y_2) \in S^2$, $U_{y_1, y_2}$ are small.
 
 \begin{lemma}
 \label{lem:fewys}
  For all but $d|S|$ pairs $(y_1, y_2) \in S^2$,
  $$U_{y_1, y_2} \leq O(d 2^{t/2}).$$
  is at most $d |S|$.
 \end{lemma}

 Assuming this for the moment, we can proceed with Equation~\eqref{eq:dbcheq}:
 \begin{align*}
  |S|^4 &\leq \left( d |S| \cdot 2^t + |S|^2 \cdot O(d \cdot 2^{t/2}) \right) \cdot |S|\\
  &= d |S|^2 2^t + O(d |S|^3 2^{t/2}).
 \end{align*}
 Thus:
 $$|S|^2 \leq d 2^t + O(d |S| 2^{t/2}),$$
 which implies that $|S| \leq O\left(d \cdot 2^{t/2} \right)$, as desired.

\end{proof}

\subsubsection*{Proof of Lemma~\ref{lem:fewys}}
\begin{proof}
Theorem~\ref{thm:weil} only applies to polynomials with odd degree. 
We first recall a standard trick involving the $\Tr$ map to deduce consequences for arbitrary degree polynomials.

Note that $\Tr(a^2) = \Tr(a)$ for all $a \in \F_{2^t}$, and that every element of $\F_{2^t}$ has a square root.
Thus for any positive degree monomial $M(X) = \alpha X^{i}$, where $i = j \cdot 2^k$ with $j$ odd, the equality:
$$ \Tr(M(x)) = \Tr(\tildeM(x))$$
for each $x \in \F_{2^t}$, where $\tildeM(X)$ is the odd degree monomial given by:
$$ \tildeM(X) = \alpha^{1/2^k} X^j.$$

Extending by linearity, this allows us to associate, to every polynomial $P(X) \in \F_{2^t}[X]$, a polynomial $\tildeP(X)$ with $$ \Tr(P(x)) = \Tr(\tildeP(x))$$ for all $x \in \F_{2^t}$, and where every monomial of $\tildeP(X)$ (except possibly the constant term) has odd degree.

The key observation is that whenever $\tildeP_{y_1, y_2}(X)$ is nonconstant, it has odd degree, and so we can apply the Weil bound.
In this case, since:
$$ \Tr(P_{y_1, y_2}(x) ) = \Tr(\tildeP_{y_1,y_2}(x))$$
for each $x \in \F_{2^t}$,
we get:
\begin{align}
U_{y_1, y_2} &=  \left|\sum_{x \in \F_{2^t}} (-1)^{\Tr(\tildeP_{y_1,y_2}(x))}  \right|\\
&\leq O(d \cdot 2^{t/2}),
\label{Ubound}
\end{align} 
where the last step follows from the Weil bound (Theorem~\ref{thm:weil}).

Thus, we simply need to show that there are at most $d |S|$ pairs $(y_1, y_2) \in S^2$ for which $\tildeP_{y_1, y_2}(X)$ is a constant. 

Suppose $f(X)$ has degree exactly $2e+1$. 
Let $\alpha$ be the coefficient of $X^{2e+1}$ in $f(X)$.
 
Define $\gamma_i(Y) \in \F_q[Y]$ by:
$$f(X+Y) = \sum_{j=0}^{2e+1} \gamma_i(Y) X^i.$$
Note that $\deg(\gamma_i(Y)) \leq 2e+1 - i$.
It is easy to check that $\gamma_{2e+1}(Y) = \alpha$
and $\gamma_{2d}(Y) = \alpha Y$.
 
Then we have:
\begin{align*}
P_{y_1, y_2}(X) &= f(X+y_1) - f(X+y_2)\\
&= \sum_{i \leq 2e} \left(\gamma_i(y_1) - \gamma_i(y_2) \right) X^i.
\end{align*}

Then, by definition,
 $$\tildeP_{y_1, y_2}(X) = \sum_{\substack{j \leq 2e-1\\ j\mbox{ odd }}} \left(\sum_{\substack{k \geq 0 \\ j2^k \leq 2e}} \left(\gamma_{j \cdot 2^k}(y_1) - \gamma_{j \cdot 2^k}(y_2)\right)^{\frac{1}{2^k}} \right) X^{j}.$$

We are trying to show that this is nonconstant for most $y_1, y_2$.
We will do this by identifying a monomial of positive degree, which often has a nonzero coefficient.
Let $e = j_0 \cdot 2^{k_0}$ with $j$ odd. We will focus on the coefficient of $X^{k_0}$.
It equals:
$$\left(\gamma_{2e}(y_1) - \gamma_{2e}(y_2) \right)^{1/{2^{k_0 +1}}} + \left(\sum_{0 \leq k \leq k_0} \left(\gamma_{j \cdot 2^k}(y_1) - \gamma_{j \cdot 2^k}(y_2)\right)^{\frac{1}{2^k}} \right) . $$

By linearity of the map $a \mapsto a^{1/2^k}$, this can be expressed in the form $Q(y_1^{1/2^{k_0+1}}, y_2^{1/2^{k_0+1}})$, where $Q(Z_1, Z_2)$ is a bivariate polynomial of degree at most $2e$. 
Furthermore, using the fact that $\gamma_{2e}(Y) = \alpha Y$, the homogeneous part of $Q(Z_1, Z_2)$ of degree $1$ exactly equals: $$\alpha^{1/2^{k_0 +1}} (Z_1 - Z_2),$$ which is nonzero. 
Thus $Q(Z_1, Z_2)$ is a nonzero polynomial.

Thus by the Schwartz-Zippel lemma, for $T = \{ y^{1/2^{k_0+1}} \mid y \in S \}$, there are at most $2e |T| \leq d |S|$ values of $(z_1, z_2) \in T^2$ such that $Q(z_1, z_2) = 0$. 
Thus there are at most $d|S|$ values of $(y_1, y_2) \in S^2$ for which the coefficient of $X^{k_0}$ in $\tildeP_{y_1, y_2}(X)$ is $0$. 
Whenever it is nonzero, Equation~\eqref{Ubound} bounding $U_{y_1, y_2}$ applies, and we get the desired result.

 \end{proof}

\section*{Acknowledgements}

We are grateful to Mike Saks, Shubhangi Saraf, and Pat Devlin for valuable discussions. We thank Alessandro Neri for very valuable remarks and references about rank metric codes and their connection to error-correcting graph codes.
We thank Lev Yohananov and Eitan Yaakobi for bringing~\cite{yohananovCodesGraphErasures2019,yohananovDoubleTripleNodeErasureCorrecting2020} to our attention.

\printbibliography

@misc{Neri-pc,
author={Neri, Alessandro},
howpublished="Personal Communication"}

@article{GowQuinlan,
  title={Galois extensions and subspaces of alternating bilinear forms with special rank properties},
  author={Gow, Rod and Quinlan, Rachel},
  journal={Linear algebra and its applications},
  volume={430},
  number={8-9},
  pages={2212--2224},
  year={2009},
  publisher={Elsevier}
}

@article{DG-alternating,
  title={Alternating bilinear forms over GF (q)},
  author={Delsarte, Philippe and Goethals, Jean-Marie},
  journal={Journal of Combinatorial Theory, Series A},
  volume={19},
  number={1},
  pages={26--50},
  year={1975},
  publisher={Elsevier}
}

@article{AlonGC1,
  title={Graph-codes},
  author={Alon, Noga},
  journal={arXiv preprint arXiv:2301.13305},
  year={2023}
}

@misc{Pat-comm,
  author = "Devlin, Pat",
  howpublished = {Personal Communication}
}

@article{AlonGC2,
  title={Connectivity Graph-Codes},
  author={Alon, Noga},
  journal={arXiv preprint arXiv:2308.07653},
  year={2023}
}

@article{AGKMS-graph-codes,
  title={Structured codes of graphs},
  author={Alon, Noga and Gujgiczer, Anna and K{\"o}rner, J{\'a}nos and Milojevi{\'c}, Aleksa and Simonyi, G{\'a}bor},
  journal={SIAM Journal on Discrete Mathematics},
  volume={37},
  number={1},
  pages={379--403},
  year={2023},
  publisher={SIAM}
}

@techreport{forney1965concatenated,
  title={Concatenated Codes.},
  author={Forney, G David},
  year={1965},
  institution={Massachusetts Institute of Technology, Research Laboratory of Electronics}
}

@article{wolfCodesDerivableTensor1965,
    title = {On Codes Derivable from the Tensor Product of Check Matrices},
    author = {Wolf, J.},
    date = {1965-04},
    journaltitle = {IEEE Transactions on Information Theory},
    volume = {11},
    number = {2},
    pages = {281--284},
    issn = {1557-9654},
    doi = {10.1109/TIT.1965.1053771},
    eventtitle = {{{IEEE Transactions}} on {{Information Theory}}},
}

@inproceedings{Schmidt1976EquationsOF,
  title={Equations over Finite Fields: An Elementary Approach},
  author={Wolfgang M. Schmidt},
  year={1976},
  url={https://api.semanticscholar.org/CorpusID:119017517}
}

@article{reedPolynomialCodesCertain1960,
  title = {Polynomial {{Codes Over Certain Finite Fields}}},
  author = {Reed, I. S. and Solomon, G.},
  date = {1960-06},
  journaltitle = {Journal of the Society for Industrial and Applied Mathematics},
  shortjournal = {Journal of the Society for Industrial and Applied Mathematics},
  volume = {8},
  number = {2},
  pages = {300--304},
  issn = {0368-4245, 2168-3484},
  doi = {10.1137/0108018},
  url = {http://epubs.siam.org/doi/10.1137/0108018},
  urldate = {2022-07-04},
  langid = {english},
}

@article{austrinSimpleDeterministicReduction2014,
  title = {A {{Simple Deterministic Reduction}} for the {{Gap Minimum Distance}} of {{Code Problem}}},
  author = {Austrin, Per and Khot, Subhash},
  date = {2014-10},
  journaltitle = {IEEE Transactions on Information Theory},
  volume = {60},
  number = {10},
  pages = {6636--6645},
  issn = {1557-9654},
  doi = {10.1109/TIT.2014.2340869},
  url = {https://ieeexplore.ieee.org/document/6868217},
  urldate = {2023-11-10},
  eventtitle = {{{IEEE Transactions}} on {{Information Theory}}}
}

@techreport{massey1963threshold,
  title={Threshold Decoding},
  author={Massey, James L},
  year={1963},
  institution={Massachusetts Institute of Technology, Research Laboratory of Electronics}
}

@article{justesenClassConstructiveAsymptotically1972,
  title = {Class of Constructive Asymptotically Good Algebraic Codes},
  author = {Justesen, J.},
  date = {1972-09},
  journaltitle = {IEEE Transactions on Information Theory},
  volume = {18},
  number = {5},
  pages = {652--656},
  issn = {1557-9654},
  doi = {10.1109/TIT.1972.1054893},
  url = {https://ieeexplore.ieee.org/document/1054893},
  urldate = {2023-11-11},
  eventtitle = {{{IEEE Transactions}} on {{Information Theory}}}
}

@article{plotkinBinaryCodesSpecified1960,
  title = {Binary Codes with Specified Minimum Distance},
  author = {Plotkin, M.},
  date = {1960-09},
  journaltitle = {IRE Transactions on Information Theory},
  volume = {6},
  number = {4},
  pages = {445--450},
  issn = {2168-2712},
  doi = {10.1109/TIT.1960.1057584},
  url = {https://ieeexplore.ieee.org/document/1057584},
  urldate = {2023-11-11},
  eventtitle = {{{IRE Transactions}} on {{Information Theory}}}
}

@article{yohananovCodesGraphErasures2019,
  title = {Codes for {{Graph Erasures}}},
  author = {Yohananov, Lev and Yaakobi, Eitan},
  date = {2019-09},
  journaltitle = {IEEE Transactions on Information Theory},
  volume = {65},
  number = {9},
  pages = {5433--5453},
  issn = {1557-9654},
  doi = {10.1109/TIT.2019.2910040},
  url = {https://ieeexplore.ieee.org/document/8685206/?arnumber=8685206},
  urldate = {2024-07-16},
  eventtitle = {{{IEEE Transactions}} on {{Information Theory}}}
}

@article{yohananovDoubleTripleNodeErasureCorrecting2020,
  title = {Double and {{Triple Node-Erasure-Correcting Codes Over Complete Graphs}}},
  author = {Yohananov, Lev and Efron, Yuval and Yaakobi, Eitan},
  date = {2020-07},
  journaltitle = {IEEE Transactions on Information Theory},
  volume = {66},
  number = {7},
  pages = {4089--4103},
  issn = {1557-9654},
  doi = {10.1109/TIT.2020.2971997},
  url = {https://ieeexplore.ieee.org/document/8985404/?arnumber=8985404},
  urldate = {2024-07-16},
  eventtitle = {{{IEEE Transactions}} on {{Information Theory}}}
}

@article{schmidtSymmetricBilinearForms2015,
  title = {Symmetric Bilinear Forms over Finite Fields with Applications to Coding Theory},
  author = {Schmidt, Kai-Uwe},
  date = {2015-09-01},
  journaltitle = {Journal of Algebraic Combinatorics},
  shortjournal = {J Algebr Comb},
  volume = {42},
  number = {2},
  pages = {635--670},
  issn = {1572-9192},
  doi = {10.1007/s10801-015-0595-0},
  url = {https://doi.org/10.1007/s10801-015-0595-0},
  urldate = {2024-07-16},
  langid = {english}
}

@inproceedings{chenRecursiveMechanismNode2013,
  title = {Recursive Mechanism: Towards Node Differential Privacy and Unrestricted Joins},
  shorttitle = {Recursive Mechanism},
  booktitle = {Proceedings of the 2013 {{ACM SIGMOD International Conference}} on {{Management}} of {{Data}}},
  author = {Chen, Shixi and Zhou, Shuigeng},
  date = {2013-06-22},
  series = {{{SIGMOD}} '13},
  pages = {653--664},
  publisher = {Association for Computing Machinery},
  location = {New York, NY, USA},
  doi = {10.1145/2463676.2465304},
  url = {https://doi.org/10.1145/2463676.2465304},
  urldate = {2024-07-18},
  isbn = {978-1-4503-2037-5}
}

@inproceedings{hayAccurateEstimationDegree2009,
  title = {Accurate {{Estimation}} of the {{Degree Distribution}} of {{Private Networks}}},
  booktitle = {2009 {{Ninth IEEE International Conference}} on {{Data Mining}}},
  author = {Hay, Michael and Li, Chao and Miklau, Gerome and Jensen, David},
  date = {2009-12},
  pages = {169--178},
  issn = {2374-8486},
  doi = {10.1109/ICDM.2009.11},
  url = {https://ieeexplore.ieee.org/document/5360242},
  urldate = {2024-07-18},
  eventtitle = {2009 {{Ninth IEEE International Conference}} on {{Data Mining}}}
}

@online{jainTimeAwareProjectionsTruly2024,
  title = {Time-{{Aware Projections}}: {{Truly Node-Private Graph Statistics}} under {{Continual Observation}}},
  shorttitle = {Time-{{Aware Projections}}},
  author = {Jain, Palak and Smith, Adam and Wagaman, Connor},
  date = {2024-03-07},
  eprint = {2403.04630},
  eprinttype = {arXiv},
  eprintclass = {cs},
  doi = {10.48550/arXiv.2403.04630},
  url = {http://arxiv.org/abs/2403.04630},
  urldate = {2024-07-18},
  pubstate = {prepublished}
}

@inproceedings{kasiviswanathanAnalyzingGraphsNode2013,
  title = {Analyzing Graphs with Node Differential Privacy},
  booktitle = {Proceedings of the 10th Theory of Cryptography Conference on {{Theory}} of {{Cryptography}}},
  author = {Kasiviswanathan, Shiva Prasad and Nissim, Kobbi and Raskhodnikova, Sofya and Smith, Adam},
  date = {2013-03-03},
  series = {{{TCC}}'13},
  pages = {457--476},
  publisher = {Springer-Verlag},
  location = {Berlin, Heidelberg},
  doi = {10.1007/978-3-642-36594-2_26},
  url = {https://doi.org/10.1007/978-3-642-36594-2_26},
  urldate = {2024-07-18},
  isbn = {978-3-642-36593-5}
}

\newpage

\section{Appendix}

\subsection{Justensen-like code} 

The construction in this example is inspired by the Justensen code \cite{justesenClassConstructiveAsymptotically1972}, which uses an ensemble of codes for the inner code instead of a single inner code. Justensen uses an ensemble known as the Wozencraft Ensemble \cite{massey1963threshold}  with the following properties.

\begin{theorem}[Wozencraft Ensemble]\label{thm:wozencraft}
    For every large enough $k$, there exists codes $C^{(1)}, C^{(2)},...,C^{(2^{k}-1)}$ over $\F_2^{2k}$ with rate $1/2$, where $1-\epsilon$ fraction of them have distance at least $H^{-1}_2(1/2 - \epsilon)$.
\end{theorem}

Since our goal is graph distance, we use the $\STCZD$ operation to convert the Wozencraft Ensemble from codes over strings with good Hamming distance to codes of matrices with good graph distance.

\begin{lemma}[Wozencraft Ensemble Modification]\label{lem:wozencraft-mod}
For any $\epsilon > 0$, and large enough $k$, there exists codes $D^{(1)}, D^{(2)},...,D^{(2^{k}-1)}$ over $\F_2^{2k \times 2k}$. View these as directed graph codes. Then these codes have rate $1/8$, and at least a $1-\epsilon$ fraction of them have distance at least $H^{-1}_2(1/2 - \epsilon)$.
\end{lemma}

\begin{proof}
    Let $C^{(1)}, C^{(2)},...,C^{(N)}$ be the Wozencraft Ensemble. For each $I \in [N]$, define $D^{(I)} = \STCZD(C^{(I)})$. Note that each $D^{(I)}$ is a code over $\F_2^{2k \times 2k}$. Note that by lemma \Cref{lem:symmetric-tensor-properties}, each of the codes has rate $1/8$. Since the $\STCZD$ operation translates Hamming distance to directed graph distance, we also have the same guarantee as the original Wozencraft Ensemble - at least $(1-\epsilon)$ fraction of the codes have distance at least $H^{-1}_2(1/2 - \epsilon)$.
\end{proof}

Concatenating $\STCZD(RS)$ with the modified Wozencraft Ensemble in a particular arrangement yields our next construction.

\begin{code}[Justensen-like $C_{\mathrm{Justensen}}(\epsilon, k, \rho)$]
Require $\epsilon, \rho \in (0,1)$.
    Let $Q = 2^k$, and $N = 2^k-1$.

    Let $D^{(1)}, D^{(2)},...,D^{2^k-1}$ be the modified Wozencraft Ensemble \Cref{lem:wozencraft-mod}.

    Then $C_{\mathrm{Justensen}}$ is the code where for each element of $A \in \STCZD(\mathrm{RS}(N, \rho, Q))$, for each $I, J \in [N]$, we replace the symbol at $A_{IJ}$ with its encoding under $D^{(\min(I, J))}$. If $J < I$, we transpose the encoding (to keep the matrix symmetric).
\end{code}

\Cref{fig:Justensen} shows where each inner code is applied.

\begin{figure}[ht]
    \centering
$$ \begin{bmatrix}
- & D^{(1)} & D^{(1)} & D^{(1)} & D^{(1)} & D^{(1)} \\
D^{(1)T} & - & D^{(2)} & D^{(2)} & D^{(2)} & D^{(2)} \\
D^{(1)T} & D^{(2)T} & - & D^{(3)} & D^{(3)} & D^{(3)} \\
D^{(1)T} & D^{(2)T} & D^{(3)T} & - & D^{(4)} & D^{(4)} \\
D^{(1)T} & D^{(2)T} & D^{(3)T} & D^{(4)T} & - & D^{(5)} \\
D^{(1)T} & D^{(2)T} & D^{(3)T} & D^{(4)T} & D^{(5)T} & - 
\end{bmatrix}  $$
    \caption{Inner code arrangement for $C_{\mathrm{Justensen}}$}
    \label{fig:Justensen}
\end{figure}

\begin{theorem}\label{thm:justensen-like}
For any $\epsilon, \rho \in (0,1)$, and $k$, a sufficiently large integer, $C_{\mathrm{Justensen}}(\epsilon, \rho, k)$ is a strongly explicit linear graph code with rate $\rho^2 / 8 - o(1)$, and distance at least $(1 - \rho - \epsilon)H^{-1}(1/2 - \epsilon)$.
\end{theorem}

\begin{proof}
Let $N = 2^k - 1$ and $n = 2k$ be the side lengths of the inner and outer codes, respectively. First note that $C_{\mathrm{Justensen}}$ is a linear graph code over $\F_2^{nN \times nN}$, since both the inner and outer codes are linear, and we can apply a $\F_2$ linear map from $\F_{2^k} \to \F_{2}^k$ before encoding with the inner code.

\textbf{Rate.} By \Cref{lem:symmetric-tensor-properties}, the outer code, $\STCZD(\mathrm{RS}(N, \rho, Q))$, has rate $\rho^2/2-o(1)$, and by \Cref{lem:wozencraft-mod}, the inner codes have rate $1/8-o(1)$. Thus, the rate is $\rho^2/8 - o(1)$ as an undirected graph code.

\textbf{Distance.} Let $O$ be a non-zero outer codeword. For convenience, let $d = H^{-1}(1/2 - \epsilon)$, and $n = 2k$ be the side length of the inner code. We claim the distance is at least $(1 - \rho - \epsilon)d$. 

Call $I \in [N]$, bad if the distance of $D^{(I)} < d$, and good otherwise. Let $B\subset [N]$ be the subset of bad indices. By the guarantee of the Wozencraft ensemble, we know that $|B| < \epsilon N$. Since $O_{IJ}$ gets encoded with $\min(I, J)$, if $I, J \notin B$, then $O_{IJ}$ is encoded with an inner code of distance at least $d$. 

Define $S_I, T_J$ as in the proof of \Cref{lem:concat}. Let $S_\geq = \{I : |S_I| \geq d_{in} \text{ and } I \notin B\}$. Similarly, define $T_{\geq}$. Then $|S_\geq|, |T_\geq| < (1 - \rho - \epsilon) N$. Then $|S_\geq \cup B|, |T_\geq \cup B < (1 - \rho)N$. Since this is less than the outer distance of the code, we have that $O_{IJ} \neq 0$ for some $I \notin S_\geq \cup B$, and $J \notin T_\geq \cup B$. In other words, $|S_I| < d$, $|T_J| < d_{in}$, and $O_{IJ}$ is encoded with a code of directed graph distance at least $d$. Thus, by the inner distance, there remains a non-zero element in the $(I, J)$th block outside of $S_{I}$, and $T_J$.
\end{proof}

\end{document}